\documentclass[runningheads]{llncs}

\usepackage{enumitem}
\usepackage{cite}
\usepackage{listings}
\usepackage{tcolorbox}
\usepackage{amsmath, amssymb, amsfonts} 
\usepackage[ruled,vlined,linesnumbered]{algorithm2e}
\usepackage{tabularx}
\usepackage{booktabs}
\usepackage[flushleft]{threeparttable}
\usepackage{fancyvrb}
\DeclareMathOperator*{\argmax}{arg\,max}
\lstdefinestyle{mystyle}{
    language=C, 
    basicstyle=\ttfamily\small, 
    numbers=left, 
    numberstyle=\tiny\color{gray}, 
    stepnumber=1, 
    numbersep=5pt, 
    tabsize=4, 
    frame=single, 
    breaklines=true, 
    postbreak=\mbox{\textcolor{red}{$\hookrightarrow$}\space}, 
}

\newcommand{\prog}{\mathcal{P}}
\newcommand{\reach}{\mathsf{Reach}}
\newcommand{\init}{\mathsf{Init}}

\newcommand{\mutex}{\mathsf{mutex}}
\newcommand{\pmutex}{\prog_{\mutex,2}}
\newcommand{\assign}{\gets}
\newcommand{\states}{S}
\newcommand{\MissedPos}{\mathsf{missedPos}}
\newcommand{\BadNeg}{\mathsf{badNegs}}
\newcommand{\NewNegs}{\mathsf{newNegs}}
\newcommand{\reached}{\mathsf{reached}}

\newcommand{\negs}{\mathsf{speculated}}
\newcommand{\program}{\mathcal{P}}
\newcommand{\sampleTrace}{\mathsf{SampleTrace}}

\newcommand{\learnInvariant}{\mathsf{ReviseInvariant}}
\newcommand{\len}{\mathsf{len}}
\newcommand{\mypara}[1]{\vspace{0.3em}\noindent {\bf #1.}}
\newcommand{\myipara}[1]{\vspace{0.3em}\noindent {\em #1.}}

\begin{document}
\title{Data-Driven Template-Free Invariant Generation}

\author{Yuan Xia\inst{1} \and
        Jyotirmoy V. Deshmukh\inst{1} \and 
        Mukund Raghothaman\inst{1} \and
        Srivatsan Ravi\inst{1}}

\authorrunning{Y. Xia et al.}
\institute{University of Southern California \\
\email{\{yuanxia,jdeshmuk,raghotha,srivatsr\}@usc.edu}}
\maketitle

\renewcommand{\baselinestretch}{0.97}
\setlength{\belowcaptionskip}{-20pt}

\begin{abstract}

Automatic verification of concurrent programs faces state explosion
due to the exponential possible interleavings of its sequential
components coupled with large or infinite state spaces. An alternative
is deductive verification, where given a candidate invariant, we
establish inductive invariance and show that any state satisfying the
invariant is also safe. However, learning (inductive) program
invariants is difficult. To this end, we propose a data-driven
procedure to synthesize program invariants, where it is assumed that
the program invariant is an expression that characterizes a (hopefully
tight) over-approximation of the reachable program states. The main
ideas of our approach are: (1) We treat a candidate invariant as a
classifier separating states observed in (sampled) program traces from
those speculated to be unreachable. (2) We develop an enumerative,
template-free approach to learn such classifiers from positive and
negative examples. At its core, our enumerative approach employs
decision trees to generate expressions that do not over-fit to the
observed states (and thus generalize). (3) We employ a runtime
framework to monitor program executions that may refute the candidate
invariant; every refutation triggers a revision of the candidate
invariant. Our runtime framework can be viewed as an instance of
statistical model checking, which gives us probabilistic guarantees on
the candidate invariant.  We also show that such in some cases, our
counterexample-guided inductive synthesis approach converges (in
probability) to an overapproximation of the reachable set of states.
Our experimental results show that our framework excels in learning
{\em useful} invariants using only a fraction of the set of reachable
states for a wide variety of concurrent programs. 

\end{abstract}

\section{Introduction}
\label{sec:intro}
\myipara{Motivation} Concurrent and distributed programs are
fundamental in modern computing, powering everything from operating
systems to web services. Ensuring their correctness is paramount, as
errors can lead to system crashes, data corruption, or security
vulnerabilities.  For example, AWS FreeRTOS is a real-time operating
system designed for IoT devices.  The Over-The-Air (OTA) system,
critical for reliable and secure firmware updates in the presence of
device and communication failures, relies on a sophisticated
distributed OTA protocol. To verify its correctness, the developers
employed the formal modeling language P for modeling and correctness
checking, uncovering three bugs in the model that pointed to potential
real-world implementation issues \cite{otacase}.

There have been several efforts to model such programs using modeling
languages such as Promela \cite{holzmann1991promela}, P
\cite{desai2013p}, and Ivy\cite{padon2016ivy}, and prove their
correctness using a variety of automatic and deductive techniques
\cite{baier2008principles,apt_verification_2009}. These languages
abstract away some of the complexity induced by real-world programs,
e.g., dynamic memory allocation and management, recursive functions,
dynamic function calls, etc., but still preserve the asynchronous
communication structure and complex coordination logic inherent in
distributed protocols. Furthermore, languages such as Promela and P
are usually accompanied by an ecosystem of tools for testing,
generating program traces, and verification\cite{delgado2004taxonomy,
desai2013p, ribeiro2005model}.

In this paper, we propose a technique to automatically learn program
invariants for arbitrary concurrent/distributed programs written in a
modeling language such as Promela or P\footnote{While the techniques
developed in this paper are general and applicable to arbitrary
modeling languages, as a proof-of-concept, we apply the technique to
distributed programs written in Promela.}.  A {\em program invariant}
is a predicate over the program's state variables that evaluates to
true in every state visited during any (infinite) program execution.
The invariant learning problem is well-motivated in the context of
proving program correctness: invariants form the basis of deductive
approaches to prove safety and liveness properties of programs
\cite{baier2008principles}. Beyond this obvious goal, {\em post hoc}
learned invariants can also facilitate programmer understanding of a
distributed computation, especially if the learned invariant matches
the programmer's understanding of the set of possible reachable
states.

\myipara{Data-driven invariant synthesis.} In {\em data}-{\em driven}
techniques, instead of statically analyzing the program text or the
state space, we leverage program executions to learn invariants.  The
objective of most of the data-driven invariant learning techniques is
to learn {\em safety invariants}, i.e., invariants whose intersection
with a user-provided set of unsafe states is empty.  Data-driven
techniques employ a counterexample-guided inductive synthesis approach
coupled with the use of an SMT solver, a theorem prover, or a model
checker to check if a candidate expression is an invariant, and to
check its safety or inductiveness
\cite{si2020code2inv,ryan2019cln2inv,padhi2016data,clarke2003counterexample,garg2014ice,icedt,ezudheen2018horn,li2017automatic,padon2016ivy}

A safety invariant is (by definition) specific to the safety property
provided by the user, and the set of states satisfying the safety
invariant could be as big as the set of safe states (many of which
could be unreachable). On the other hand, the exact set of reachable
states of a program is the smallest inductive invariant  for a
program. Obtaining an expression characterizing this tight invariant
decouples the safety property from the search for the invariant; in
fact such an invariant can help prove several distinct safety
properties.  The main question we ask in this paper is: {\em Can we
learn a (compact) expression over the program variables that is
satisfiable only by a ``tight'' overapproximation of the set of
reachable states?}. Such an expression can\footnote{Of course, this
depends on the actual set of unsafe states.} potentially serve as a
safety invariant, and can also provide a nice summary of the
distributed computation that has value beyond a specific verification
exercise. 

\myipara{Learning tight inductive invariants: Challenges and Approach}
Following Occam's razor principle in learning theory \cite{razor}, we
prefer invariant expressions that are shorter in length. However, this
poses an interesting challenge: In learning safety invariants,
observed program states are used as positive examples, unsafe states
as negative examples, and invariant learning is the well-framed
problem of learning a classifier. As we do not assume a specific
safety property, learning shortest invariant expressions may give
trivial answers, after all {\em true} is a valid invariant (in the
absence of any negative states)! To overcome this hurdle, we propose
the novel use of {\em speculative negative examples}.  We sample a set
of program executions; states encountered in these executions are
treated as positive examples, and from the set of heretofore
unobserved states, we randomly speculate some states to be
unreachable.  Now equipped with this set of positive and negative
examples, the problem of learning a classifier for these sets is again
well-defined. We treat the expression thus learned as a {\em candidate
invariant}, and next discuss how we can validate its invariance.

A key limitation in existing data-driven techniques is their
dependence on model checking or theorem proving to establish validity,
inductiveness, and safety of the invariant. To sidestep this hurdle,
we propose a statistical/runtime monitoring approach to get
probabilistic guarantees about the learned invariant. We simply
continue running the program or sampling its executions (assuming that
the scheduler picks different program executions with some non-zero
probability). If an execution refutes the candidate (by reaching a
state that violates the candidate or reaching a state speculated to be
unreachable), our method triggers a revision of the candidate
invariant. We make a probabilistic argument that for certain programs,
frequency of such revision events goes to zero, or that this procedure
converges to a (tight) over-approximation of the reachable states. In
the general case, runtime sampling of bounded-length program
executions is an instance of statistical model checking
\cite{agha2018survey,smc}. For a given probability threshold
$\epsilon$, we can compute the number of program runs $n$, where if
the candidate invariant is satisfied by $n$ bounded-length program
executions, then the probability of this invariant being correct is
greater than $\epsilon$.


A final contribution of this paper is that we eschew the use of
specific templates or predicates to define the space of candidate
invariants as used by techniques in
\cite{flanagan2001houdini,ernst2007daikon}.  We assume that the user
provides a grammar to define the (possibly large) set of
Boolean-valued expressions over program variables of interest, and
that we learn an expression from this grammar that successfully
separates the positive states from the speculated negative states.
Our approach combines decision-tree learning and an enumerative,
syntax-guided synthesis \cite{sygus} approach to identify approximate
invariants.




\myipara{Contributions}
To summarize, this paper makes the following contributions:
\begin{enumerate}[nosep,leftmargin=0.3in,label={(\roman*)}]
    \item We propose an algorithmic method that efficiently generates candidate invariants without relying on templates.
    \item We efficiently learn invariants from positive and speculative negative examples using heuristic optimization procedures during the learning process.
    \item Our framework blends continuous monitoring of candidate invariants and inference of revised invariants in a seamless runtime framework.
    \item We provide statistical guarantees on the soundness of our approach.
    \item We demonstrate the practical usefulness of our approach on several example distributed programs written in the Promela language.
\end{enumerate}
%

\section{Preliminaries}
\label{sec:prelim}
\newcommand{\numproc}{N}
\newcommand{\tracelength}{k}

In this section, we formalize the terminology needed to explain our
approach with the aid of Peterson's algorithm to ensure mutually
exclusive access to the critical section for two concurrent
processes\cite{peterson1981myths} (henceforth denoted as $\pmutex$ for
brevity). We show $\pmutex$ modeled using the Promela language
\cite{holzmann1991promela} in Fig.~\ref{fig:pmutex}. 

\newcommand{\bool}{\mathtt{bool}}
\newcommand{\inte}{\mathtt{int}}
\newcommand{\byte}{\mathtt{byte}}
\newcommand{\pc}{\mathtt{pc}}
\lstset{escapeinside={<@}{@>}}
\begin{figure}[t]
\begin{lstlisting}[style=mystyle]
// Peterson's solution to the mutual exclusion problem - 1981 
bool turn, flag[2];
byte ncrit;
active [2] proctype user() {
	assert(_pid == 0 || _pid == 1);
again:
	flag[_pid] = 1;
	turn = _pid;
	(flag[1 - _pid] == 0 || turn == 1 - _pid);
    <@\textcolor{blue}{
	ncrit++;}@>
    <@\textcolor{blue}{
	assert(ncrit == 1);	// critical section}@>
    <@\textcolor{blue}{
	ncrit--;}@>
	flag[_pid] = 0;
	goto again }
\end{lstlisting}
\caption{Peterson's mutual exclusion algorithm for 2 processes in Promela}
\label{fig:pmutex}
\end{figure}

\newcommand{\typeof}[1]{\mathtt{type}(#1)}

\myipara{Program structure and modeling assumptions}
Program variables, the program state, and program execution traces
are all standard notions, we define them formally so that we can
precisely state the problem we wish to solve.

\newcommand{\val}{\nu}
\newcommand{\vars}{V}

\begin{definition}[Variables, State] A type $t$ is a finite or an
infinite set. A program variable $v$ of type $t$ is a symbolic name
that takes values from $t$. We use $\typeof{v}$ to denote variable types. We use $\vars$ to denote the set of program variables.
A {\em valuation} $\val$ maps a variable $v$ to a specific value in
$\typeof{v}$. We use $\val(\vars)$ to denote the tuple containing
valuations of the program variables; a program state is a specific
valuation of its variables during program execution.
\end{definition}

\newcommand{\integers}{\mathbb{Z}}
Types may include: $\bool = \{0,1\}$, $\inte = \integers$, $\byte =
\{0,1,2^8-1\}$, finite enumeration types which is some finite set of
values, the program counter $\pc$ type that takes values in
the set of line numbers of the program. We assume that our modeling
language has a grammar to define syntactically correct expressions
consisting of program variables and operators, and a well-defined
semantics that defines the valuation of an expression using the
valuations of the variables. Given a set of variables
$\{v_1,\ldots,v_k\}$ and an expression $e(v_1,\ldots,v_k)$, we use
$\val(e)$ to denote the valuation of the expression when each $v_i$ is
substituted by $\val(v_i)$.

Most real-world programs have procedure calls and procedure-local
variables; we assume that the program is modeled in a language that
abstracts away such details. Real-world concurrent programs typically
use either a shared memory or message-passing paradigm to perform
concurrent operations of individual {\em threads} or {\em
processes}\footnote{In most modern programming languages, processes
and threads have been used to mean different abstract units of
concurrent operation. For the purpose of this paper, we assume that we
are using a modeling language such as Promela or P that abstracts away
from these finer details.}. We assume that the number of concurrent
processes in the program, $\numproc$, is known and fixed throughout
program execution. Thus, our modeling language does not have
statements to start a new process or terminate one. So, a concurrent
program for us is a set of $\numproc$ processes, with each process
being a sequence of atomic statements. We include statements such as
assignments and conditional statements (which execute sequentially),
and $\mathtt{goto}$ statements to alter the sequential control flow.
An assignment statement is of the form $v_i \assign e$, where $e$ is
an expression whose valuations must be in $\typeof{v_i}$. The formal
semantics of an assignment statement is that for all variables that do
not appear on the LHS of the assignment, the valuation remains
unchanged, while the valuation of $v_i$ changes to $\val(e)$.

\begin{example}
\label{ex:statetran}
For $\pmutex$, $\vars = \{\mathtt{\ell,\ell_2,turn, flag, ncrit}\}$,
where $\typeof{\ell,turn,flag[i],\\ncrit} =
(\pc\times\pc,\bool,\bool\times\bool,\byte)$. The variable $\ell$ is a
pair that maintains program counters for the two processes.  The
statements $\mathtt{flag[\_pid] = 1}$ and $\mathtt{turn = \_pid}$ are
assignments. The state $s_0$ is
\[
s_0 = \left(
\begin{aligned}
\ell_1 \mapsto 7
\ell_2 \mapsto 7
\mathtt{turn} \mapsto 0,
\mathtt{flag[0]} \mapsto 0,
\mathtt{flag[1]} \mapsto 0,
\mathtt{ncrit} \mapsto 0
\end{aligned}
\right) \qquad
\]
After Process 0 ($\mathtt{\_pid} = 0$) executes two assignment
statements ($\mathtt{flag[0]} \assign 1$ and $\mathtt{turn} \assign
0$), the program state changes to the following:\[
s_{2} = \left(
\begin{aligned}
\ell_1 \mapsto 9,
\ell_2 \mapsto 7,
\mathtt{turn} \mapsto 0,
\mathtt{flag[0]} \mapsto 1,
\mathtt{flag[1]} \mapsto 0,
\mathtt{ncrit} \mapsto 0 
\end{aligned}
\right)
\qquad \qed
\]
\end{example}
The execution semantics of a program are typically explained using the
notion of a labeled transition system.

\begin{definition}[Labelled Transition System (LTS) for a concurrent
program]
A labelled transition system is a tuple $(S, L, T, \init)$ where $S$
is the state space of the program, $T \subseteq S \times S$ is the 
transition relation, $L$ is the set of program statements that serve as 
labels for elements in $T$, and $\init$ is a set of initial program 
states.
\end{definition}
Transitions specify how a system evolves from one state to another,
and labels indicate the program statement that induced
the transition. Converting a program to its corresponding
LTS is a well-defined process (see \cite{baier2008principles});
where transition is added based on the effect of the corresponding
program statement on the program variables (including the program
counter).
For assignment statements, this involves updating the program counter and
the valuation of the variables on the LHS. For conditional statements,
transitions are added to the next state only if the valuation
corresponding to the current state satisfies the condition. For
$\mathtt{goto}$ statements, only the program counter is updated.
Computing the LTS for a concurrent program is a bit more tricky. We
typically use an interleaved model of concurrency, so from every
state, we consider $\numproc$ next states corresponding to each of the
$\numproc$ concurrent processes executing.

We assume that each of the sequential
processes involved in the computation is {\em deterministic}, and the
only source of nondeterminism is context switches due to an external
scheduler\footnote{Our techniques can also handle nondeterminism in
the sequential processes as long as the number of nondeterministic
choices available is fixed and known {\em a priori}.  This restriction
is only required to give the theoretical guarantees of convergence
that we later discuss. From a practical perspective, our method can be
applied to unbounded nondeterminism.}. We also remark that an LTS is finite if the type
of each program variable is finite, otherwise an LTS can have an
infinite set of states.

\begin{definition}[Reachable States] For any program $\prog$, its set
of reachable states $\reach(\prog, \init)$ is defined as the set of
all states that can be reached from an initial state $s_0 \in S$
through the execution of the program statements.  Formally,
$\reach(\prog,\init)$ is the smallest set $R$ that satisfies the
following:
\begin{enumerate}
    \item $s_0 \in \init \implies s \in R$, and,
    \item $s \in R \land (s,s')\in T \implies s'\in R$.
\end{enumerate}
\end{definition}

\begin{example}
We remark that the state $s_2$ in Example~\ref{ex:statetran}is reachable from $s_0$, since: \(s_0
\xrightarrow{flag[\_pid]=1;turn=\_pid} s_2\). States
where $ncrit=2$ are unreachable from $s_0$. 
$\hfill  \qed$
\end{example}

\begin{definition}[Program Trace]
A program trace $\sigma$ of length $\tracelength = \len(\sigma)$ is a
sequence of states $s_0,s_1,\ldots,s_{\tracelength-1}$, s.t., $s_0\in
\init$, and for all $j \in [1,\tracelength-1]$, $(s_{j-1},s_j) \in T$.
\end{definition}

\newcommand{\sttt}[1]{\langle#1\rangle}
\newcommand{\ttt}[2]{\underbrace{\sttt{#1}}_{#2}}
\begin{example}
\label{ex:trace}
A possible trace $\sigma$ of $\pmutex$ with $\len(\sigma) = 5$ is shown below. 
\[
\left(
\begin{array}{l}
\ttt{7\!,7\!,0,\!0,\!0,\!0}{s_0},
\ttt{8\!,7\!,1,\!0,\!0,\!0}{s_1},
\ttt{9\!,7\!,1,\!0,\!0,\!0}{s_2},
\ttt{11\!,7\!,1,\!0,\!0,\!1}{s_3},
\ttt{13\!,7\!,1,\!0,\!0,\!0}{s_4},
\end{array}
\right) \qed
\]
\end{example}

\begin{remark}
To obtain a program trace, starting from a random $s_0 \in \init$, we
can randomly sample a successor $s_1$ from all possible pairs $(s_0,s')
\in T$, and repeat this procedure from each subsequent $s_i$. Some of the
successor states for a given state $s$ correspond to a context switch
for the distributed program (as it may require a different process to
execute its atomic instruction than the one that executed to reach the
state $s$). To randomly sample the initial or successor
states, we need a suitable distribution over the initial states and outgoing labelled
transitions from a given state.  This distribution is defined by the
scheduler and assumed to be {\em unknown} to the program developer. 
\end{remark}

\begin{definition}[Invariants]
An invariant $I$ is a Boolean-valued formula over program variables
and constants that is satisfied by every reachable
program state. 
\end{definition}

\begin{example}
For instance, in the Peterson's model, an invariant is $ncrit \leq 1$.
This invariant is useful to show that at most one process is in the
critical section at any given time.  $\hfill \qed$
\end{example}


\noindent {\em \underline{Problem statement}.} Let $\prog$ be a distributed program.
Let $\reach(\prog, \init)$ be the set of reachable states of $\prog$. The objective of this
paper is to design an algorithm to learn a candidate program invariant
$\phi$ that has the following properties:
\begin{enumerate}
\item 
Soundness of the invariant:
    \( (s \in \reach(\prog,\init)) \Rightarrow (s \in \phi) \) 
\item 
(Probabilistic) tightness of the invariant:
\begin{equation}
\label{eq:probtight}
    \Pr[\ (s \not\in \reach(\prog)) \wedge (s \in \phi) \ ] < \epsilon
\end{equation}
\end{enumerate}

\begin{remark}
Checking the soundness of the candidate invariant typically requires
the use of a verification tool such as a model checker or theorem
prover. This use of a verification tool is {\bf not} a part of our
invariant synthesis procedure. However, for experiments,
we use the Spin model checker to show that our candidate invariants are indeed valid program invariants and to demonstrate that we can (with high
probability) learn a valid invariant {\em without using a verification
tool}.  However, we note that our procedure {\em does not guarantee}
soundness of the synthesized invariant. Instead, the procedure can
only provide a weaker probabilistic guarantee: \mbox{\( \Pr((s \in
\reach(\prog,\init)) \wedge (s \not\in \phi)) < \epsilon\)} for some
small value of $\epsilon.$
\end{remark}

\section{Data-driven Invariant Generation}
\label{sec:solution}
\DontPrintSemicolon
\begin{algorithm}[!t]
\caption{$\mathsf{InvGen}$}\label{alg:InvSynAbs}
\SetKwInOut{Input}{input}
\SetKwInOut{Output}{output}
\Input{Distributed/Concurrent program $\program$}
\Output{A candidate invariant}
$\reached \gets \{\}$,  $\negs \gets \{\}$,  $\phi \gets \mathit{false}$\;
\Repeat{$n \le$ trace budget for statistical verification}{
    $T \gets \sampleTrace(\program)$ \nllabel{al:sample}
    \tcp*[l]{Sample a random program trace}
    $\MissedPos \gets \{s \mid s\in T \wedge s\not\models \phi\}$ \nllabel{al:newpos} 
    \tcp*[l]{Reached States not in $\phi$}
    $\reached \gets \reached \cup \MissedPos$ \tcp*[r]{Update the set of reached states} \nllabel{al:update}
    \;
    \tcc{Remove reached states from speculated unreachable states, and add new speculated unreachable states}
    $\NewNegs \gets \{s \mid s \in \mathsf{randomSample}(V,i) \wedge (s \notin \reached)\}$ \nllabel{al:newspec} \;
    $\BadNeg \gets \{s\!\mid s\in T\! \wedge\! s\in \negs \}$  \nllabel{al:badnegs} \;
    $\negs \gets (\negs \cup \NewNegs) \setminus \BadNeg$ \nllabel{al:updateneg} \tcp*[l]{Update speculated}
    \;
    \If {$\MissedPos \neq \emptyset \lor \BadNeg \neq \emptyset \lor \NewNegs \neq \emptyset$}{
        $\phi \gets \learnInvariant(\reached,\negs)$ \nllabel{al:revise} \tcp*[r]{Invariant Revision}
        $n \gets 0$ \nllabel{al:reset}
    }\Else {
        $n \gets n + 1$ \nllabel{al:incround}
    }
} 
\end{algorithm}

In this section, we present the high-level invariant synthesis
approach in Algorithm~\ref{alg:InvSynAbs}. At a high level, our
algorithm learns new invariants from randomly generated program
traces. We use the set $\reached$ to denote the set of states reached
during exploration.  The candidate invariant is denoted by $\phi$. In
Line~\ref{al:sample}, we randomly sample a new program trace of an
{\em a priori} fixed length.  The algorithm maintains three items
across all iterations: $\reached$, $\phi$, and the set of states
speculated to be unreachable, denoted by $\negs$.  Initially,
$\reached$ and $\negs$ are empty, and $\phi$ is set to
$\mathit{false}$.  If there are states in the sampled trace that are
not satisfied by the current invariant, we call them missed positive
states or $\MissedPos$ (Line~\ref{al:newpos}).  As these states are
discovered as reachable, they are added to $\reached$
(Line~\ref{al:update}). States that are speculated to be unreachable
but are discovered to be reached in the trace are called $\BadNeg$
(Line~\ref{al:badnegs}).  In every iteration we speculate some $i$ new
states ($\NewNegs$) to be unreachable, these are chosen randomly from
the set of possible program states (i.e. the Cartesian product of the
domains of program variables), but with care not to pick any state in
$\reached$.  In each iteration, we update $\negs$ to include the
states in $\NewNegs$ and remove the states in $\BadNeg$
(Line~\ref{al:updateneg}).  If either $\MissedPos$ or $\BadNeg$ are
non-empty, it means that the invariant is smaller than the set of
reachable states and needs to be revised (Line~\ref{al:revise}). The
exact procedure to learn the invariant expression from positive
($\reached$) and negative ($\negs$) state examples is explained in
Section~\ref{sec:learning}.  

Once a new invariant expression is learned, we reset the number of times the
invariant will be challenged to $0$ (Line~\ref{al:reset}). If the invariant
expression is {\em not revised} in a particular iteration, it means that
the invariant met that round's challenge (through the sampled trace),
incrementing the counter $n$ (Line~\ref{al:incround}). If the invariant
survives a user-specified number of rounds the algorithm terminates.
For example, this bound could be derived from the number of runs
required in a statistical model checking approach to get the desired
level of probability guarantee $\epsilon$ specified in
\eqref{eq:probtight}.



We now explain the detailed steps of the algorithm with the help of the
running example of $\pmutex$.

\mypara{Sampling Program Traces}
The $\sampleTrace(\cdot)$ sub-routine generates random program
executions containing a fixed number of steps. As mentioned in
Section~\ref{sec:prelim}, the sequential parts of the program are
deterministic. After each program statement executed by process $i$,
it is possible for the scheduler to pick the program statement in
process $j$ to execute next. We assume that at any step, for every
$i$, the random scheduler picks the next statement in process $i$ with
non-zero probability. 

Recall that the state of $\pmutex$ is a valuation of $\mathtt{(turn,
flag[0], flag[1],ncrit)}$. The trace of valuations of these state
variables from Example~\ref{ex:trace} is a possible random trace
of $\pmutex$ obtained by using $\sampleTrace(\pmutex)$, reproduced
here for ease of exposition:
\[
\sttt{7\!,7\!,0,\!0,\!0,\!0},
\sttt{8\!,7\!,1,\!0,\!0,\!0},
\sttt{9\!,7\!,1,\!0,\!0,\!0},
\sttt{11\!,7\!,1,\!0,\!0,\!1},
\sttt{13\!,7\!,1,\!0,\!0,\!0},
\sttt{7\!,7\!,0,\!0,\!0,\!0}
\]
Consider the first iteration of Algorithm~\ref{alg:InvSynAbs},
then the set $\MissedPos$ will be empty (as the set $\reached$ is
empty). In Line~\ref{al:update}, we will add all states from the trace
above to $\reached$. 

\mypara{Speculative Negative Examples}
In Line~\ref{al:newspec}, we speculatively add random states not in
$\reached$ to the set $\NewNegs$. Suppose we add the following
states:
\begin{equation}\label{eq:speculated}
\negs = \{\sttt{7\!,8\!,0\!,1\!,0\!,0}, \sttt{7\!,9\!,0\!,1\!,1\!,0}\}
\end{equation}
Again, since this is the first iteration of
Algorithm~\ref{alg:InvSynAbs}, the set $\BadNeg$ is empty. 

\mypara{Learning the invariant} Since the set $\NewNegs$ is not empty, we
now use the positive ($\reached$) and negative ($\negs$)
examples to learn a candidate invariant $\phi$ that overapproximates
$\reached$ but has minimal overlap with $\negs$. In other words,
we learn $\phi$ s.t. $\forall s\in\reached, s\models \phi$, and
$|\{s\mid s\models\phi \wedge s\in \negs\}|$ is less than some
threshold.  The details of the invariant learning procedure are
presented in the next section.  For the running example, suppose we
learn the invariant $\phi_1 \equiv \mathtt{flag[1]}=0$. We can check
that all states in $\reached$ satisfy $\phi_1$ and none of the states
in $\negs$ satisfy it.

\mypara{Challenging the invariant} The next step is to find sample
traces to challenge the invariant. We use at most $n$ such sample
traces in an attempt to refute the invariant, where the number $n$ is
either user-specified based on some heuristics or based on systematic
statistical reasoning required to get the desired probability
threshold $\epsilon$ \cite{smc}.  Suppose we sample the following
trace next:
\[
\sttt{7\!,7\!,0,\!0,\!0,\!0},
\sttt{7\!,8\!,0,\!1,\!0,\!0},
{\bf \sttt{7\!,9\!,0,\!1,\!1,\!0}},
\sttt{7\!,11\!,0,\!1,\!1,\!1},
\sttt{7\!,13\!,0,\!1,\!1,\!0},
\sttt{7\!,7\!,0,\!0,\!1,\!0}
\]
Clearly, the states in this trace refute the invariant $\phi_1$
because $\mathtt{flag[1]}\neq 0$, which means that these states will
appear in the set $\MissedPos$.  We note that the second state in the
set $\negs$ (shown in \eqref{eq:speculated}) was speculated to
be unreachable but is actually reached in the trace (shown in bold).
Thus, in Line~\ref{al:badnegs}, this state gets added to the set
$\BadNeg$, and is removed from $\negs$ and inserted in
$\reached$. With the revised $\negs$ and $\reached$, we can
re-learn the invariant.

\subsection{Convergence of InvGen}

In this section, we argue that under the right set of assumptions, the
$\mathsf{InvGen}$ algorithm converges to a program invariant that is
sound.  For simplicity, we assume that the distributed program has a
unique initial state $s_0$; the arguments made in this section easily
generalize to a finite number of initial states by simply adding a
dummy unique initial state that (nondeterministically) transitions to
any of them.

Next, we introduce a probabilistic interpretation of the
nondeterministic scheduler. For every set of transitions
$\{(s,s_1),\ldots,(s,s_\ell)\}$ emanating from a state $s$, we assume
that there exists a categorical probability distribution $P(s'\mid
s)$, i.e., we associate with each transition $(s_i,s_j)$ some
probability $p_{ij} = P(s'=s_j\mid s = s_i)$ s.t. $\sum_i p_{ij} = 1$,
and $\forall j, p_{ij} > 0$.
This effectively converts the program LTS into a Markov chain. The
transition probability matrix of such a Markov chain is a matrix where
the $(i,j)^{th}$ entry denotes the probability of transitioning from
state $s_i$ to state $s_j$ due to a single atomic statement
(assignment or conditional). 
We define the $k$-reach set of the program $\prog$ as the set of all
states that can be reached from $s_0$ with at most $k$ transitions.
Formally,
\(
\reach^k(\prog,s_0) = \left\{ s \middle| (s_0,s) \in \cup_{j=1}^{k} T^j \right\}
\).
We note that irrespective of the size of the state space (finite or
infinite), $\reach^k(\prog,s_0)$ for any distributed program $\prog$
with a finite number of concurrent processes is finite.

We use $\sigma \sim \prog$ to denote a trace sampled from $\prog$
using the probabilistic interpretation for the scheduler. Formally, if
$\sigma = (s_0, s_1, \ldots, s_{k-1})$, then for each $j \in [1,k-1]$,
$s_j \sim P(s'\mid s=s_{j-1})$.  Let $\states(\sigma)$ be abuse of
notation to denote the set of states present in the trace $\sigma$.

\begin{lemma}
\label{lem:q}
Let $Q = \states(\sigma_1) \cup \ldots \cup \states(\sigma_m)$, i.e.,
the set of states visited after $m$ $k$-length random traces $\sigma_j$ sampled 
from $\prog$. Then, there exists a finite $m$ s.t. $Q$ almost surely 
contains $\reach^k(\prog,s_0)$.
\end{lemma}

\noindent The proof relies on the assumption that every nondeterministic transition has
a non-zero probability. Thus, for every state in $\reach^k(\prog,s_0)$, the
probability of reaching it from $s_0$ is non-zero. From this, it follows that
if there are a sufficiently large number of sample traces picked, every state
in $\reach^k(\prog,s_0)$ must be eventually visited.
Next, we observe that in the $n^{th}$ revision round, the $\mathsf{InvGen}$ algorithm (Alg.~\ref{alg:InvSynAbs}) guarantees that the invariant $\phi_n$ learned from the set
of visited states ($\reached_n$) and speculative unreachable states
($\negs_n$) satisfies that $\phi_n \Rightarrow \reached_n$, i.e. $\phi_n$ is
strictly larger than the set of positive states. Now, we 
state our theorem, which establishes that the learned invariants are
tight.

\begin{theorem}
For an arbitrary distributed program $\prog$ with a finite number of
concurrent processes, and a procedure to sample $k$-length traces of
$\prog$ with a scheduler that is guaranteed to pick every
nondeterministic transition with non-zero probability,
Algorithm~\ref{alg:InvSynAbs} identifies an invariant expression
$\phi_\ell$ satisfying the property: $(s \in \phi_\ell) \Rightarrow
(s \in \reach^k(\prog,s_0))$ after $\ell$ revisions.
\end{theorem}
\begin{proof}
The proof sketch is as follows:
    Firstly, after any revision round $n < \ell$, we show that $|\reached_n| <
    |\reached_{n+1}|$, i.e.  the set of visited states strictly increases
    as long as InvGen uses a large enough number of rounds before
    declaring the candidate invariant as the final synthesized
    invariant.
    Secondly, we argue that if $\reached_{n}$ strictly increases in size, as
    $\reach^k(\prog,s_0)$ is a finite set, eventually $\reached_n$ must
    exceed $\reach^k(\prog,s_0)$ in size, guaranteeing the property
    indicated in the statement of the theorem.

For the first part of the proof, we observe that an invariant revision
happens if a new state $s$ is visited such that $s \not\models \phi$
or $s \in \negs$. From Lemma~\ref{lem:q}, there is a finite number of
rounds after which almost surely every $k$-step reachable state is
visited. Thus, if there is a state not in the candidate invariant but
in the $k$-reach set, it will be visited, triggering a revision.
Moreover, this state will be added to $\reached_n$, i.e.
$\reached_{n+1}$ will be strictly larger than $\reached_n$. \hfill $\qed$

\end{proof}

\begin{corollary}
Algorithm~\ref{alg:InvSynAbs} can always find an invariant that
includes the true reachable set of states for a finite-state
distributed program.
\end{corollary}
\begin{proof}
The proof follows from the fact that there is a $k$ for every
finite-state program s.t. the $k$-reach set is a least fixed point, or
the true reachable set of states. Thus, with a large enough $k$
(diameter of the state space), Algorithm~\ref{alg:InvSynAbs} will
always find an invariant containing the true reachable set of states. \hfill $\qed$
\end{proof}

\section{Decision Tree Learning for Invariants}
\label{sec:learning}
In this section, we will explain the details of the $\learnInvariant$
method in Algorithm~\ref{alg:InvSynAbs}. The $\learnInvariant$ algorithm combines decision tree
learning, syntax-guided enumerative learning, and a novel sub-sampling
approach that helps facilitate rapid inference of invariants.
%
%
\newcommand{\grammar}{G}
\newcommand{\atom}{\mathsf{atom}}
\newcommand{\atoms}{\mathsf{atoms}}
\newcommand{\param}{\mathbf{p}}
\newcommand{\precision}{\mathsf{precision}}
\newcommand{\recall}{\mathsf{recall}}

\mypara{Syntax-guided Synthesis} 
Our main idea for learning invariant expressions is drawn from
syntax-guided synthesis \cite{sygus}. We assume that we are given
a grammar $\grammar$ that defines Boolean-valued formulas over {\em parametric atomic predicates}:
\(
\phi ::= \atom(\param) \mid \neg \phi \mid \phi \wedge \phi \mid \phi \vee \phi.
\)
Here, $\atom(\param)$ expressions are user-supplied,
Boolean-valued expressions over the program variables and parameters that we call parametric atomic predicates.
Parameters $\param$ in $\atom$ are placeholders for constant values
of corresponding types. Replacing all parameters with appropriate
constants gives a predicate symbol over the program variables.
\begin{example}
Consider the following grammar for parameterized atom symbols.
\[
\atom(c) ::= (x \le c) \mid (x \ge c) \mid (y \le c) \mid (y \ge c) 
\]
Here, the program variables $\vars$ is the set $\{x,y\}$ (say of type
$\byte$), and $c$ is a parameter of type $\byte$. Substituting $c$ with
values, e.g., 2, -3, etc., gives atomic predicates $x \le 2$, $y \ge
-3$, etc. $\hfill \qed$
\end{example}

\newcommand{\subP}{\tilde{P}}
\newcommand{\subN}{\tilde{N}}

\begin{algorithm}[t]
\caption{$\mathsf{DecisionTreeLearner}(P,N,\delta)$\label{alg:dtree}}
\SetKwInOut{Input}{input}
\SetKwInOut{Output}{output}
\Input{\begin{itemize}
    \item 
        $P$ : set of positive examples, 
        $N$: set of negative examples, 
    \item $\delta$ : precision threshold, 
    \item $k$ : bound on size of $P,N$ sets used to perform $\mathsf{InvLearn}$ 
\end{itemize}}
\Output{A candidate invariant $\phi$}
\If{$|P|+|N| > k$}{ \nllabel{al:bound}
    $\atom$ = $\displaystyle\argmax_{\atom \in \atoms} (\precision(\atom,P,N))$ \; \nllabel{al:entropy}
    $P' \gets \{ s \mid s \in P, s \models \atom\}$,
    $N' \gets \{ s \mid s \in N, s \not\models \atom\}$ \;
    $\phi \gets (\atom \wedge \mathsf{DecisionTreeLearner}(P',N',\delta) \ \ \vee$ \\ 
    \qquad $(\neg\atom \wedge \mathsf{DecisionTreeLearner}(P-P',N-N',\delta)$\; \nllabel{al:dtlearner}
} \lElse {
    $\phi \gets \mathsf{InvLearn}(P,N,\delta)$ \nllabel{al:invlearn}
}
\Return $\phi$ 
\end{algorithm}

\noindent Next, we define the notion of a signature of a well-formed formula in
$G$. We assume that we are given an
ordered set of positive examples $P = \langle
e_1,\ldots,e_{|P|}\rangle$ and an ordered set of negative examples $N
= \langle f_1,\ldots,f_{|N|}\rangle$. 

\newcommand{\chk}{\mathsf{Check}}
\begin{algorithm}[t]
\caption{$\mathsf{InvLearn}(P,N,\delta)$}\label{alg:InvRevision}
\SetKwInOut{Input}{input}
\SetKwInOut{Output}{output}
\SetKw{Break}{break}
\Input{$P$ : set of positive examples, $N$: set of negative examples, $\delta$ : precision threshold}
\Output{A candidate invariant $\phi$}
$\subP,\subN$ = $\mathsf{subSample}(P,N)$ \nllabel{al:subsample} \;
let $\chk(\psi)$ = $\lambda.\psi\ (\precision(\psi,\subP,\subN) > \delta\ \wedge\ \recall(\psi,\subP,\subN) = 1)$ \;
\While{number of sub-sampling rounds not exceeded}{
    $\ell \gets 1$ \;
    \While{$\ell \leq \mathsf{maxInvLength}$}{
        \ForEach{$\psi$ in $\atoms$}{
            \lIf{$\chk(\psi)$}{
                \Break \nllabel{al:breakatom}
            } \lElse {
                $\Phi \gets \Phi \cup \{ \psi \}$
            }
        }
        \If{$\ell > 1$}{
            \ForEach{$\varphi \in \Phi$ s.t. $\len(\varphi) = \ell-1$}{ 
                $\psi \gets \neg\varphi$ \label{al:neg} \;
                \lIf{$\chk(\psi)$}{
                    \Break \nllabel{al:breakneg}
                } \lElse {
                    $\Phi \gets \Phi \cup \{ \psi \}$ 
                }
            }
            \ForEach{$\varphi_1,\varphi_2$ in $\Phi$ s.t. $\len(\varphi_1) + \len(\varphi_2) = \ell-1$}{
                \ForEach{$\circ \in \{\wedge,\vee,\Rightarrow,\equiv,\ldots\}$}{ 
                    $\psi \gets \varphi_1 \circ \varphi_2$ \label{al:op} \;
                    \lIf{$\chk(\psi)$}{
                        \Break \nllabel{al:breakop} 
                    } \lElse {
                        $\Phi \gets \Phi \cup \{ \psi \}$
                    }
                }
            }
        }
        $\ell \gets \ell + 1$ \nllabel{al:inc}
    }
    \lIf{$\precision(\psi,P,N) > \delta \wedge \recall(\psi,P,N) = 1$}{
        \Return $\psi$ \nllabel{al:return}
    } \lElse {
        $\subP = \subP \cup \mathsf{subSample}(P), \subN = \subN \cup \mathsf{subSample}(N)$ \nllabel{al:resubsample}
    } 
}
\end{algorithm}

\begin{definition}[Formula signature]
Given ordered sets $P$ and $N$, the signature of a formula $\phi$ is
an $(|P|+|N|)$-bit vector $\sigma^\phi$, where for each $i \in [\,1,|P|\,]$ ,
$\sigma^\phi_i = 1$ \emph{iff} $e_i \models \phi$ and $0$ otherwise, and for $i \in [\,|P|+1,|P|+|N|\,]$, $\sigma^\phi_i = 1$ iff $f_{i-|P|} \models \phi$ and $0$ otherwise.
\end{definition}
In plain words, assuming the sets $P$ and $N$ are concatenated, the
bit at a given index in $\sigma$ is $1$ \emph{iff} the corresponding example
in the concatenated set satisfies $\phi$.

We now observe that given the formula signatures for concrete atomic
predicates, it is trivial to recursively compute the signatures for
arbitrary formulas in the grammar using the following rules, where $\mathsf{bw\_op}$ represents the bitwise
application of the corresponding operation.
\[
\begin{array}{lll}
\sigma^{\neg \phi} = \mathsf{bw\_not}(\sigma^\phi), \quad& 
\sigma^{\phi_1 \wedge \phi_2} = \mathsf{bw\_and}(\sigma^{\phi_1},\sigma^{\phi_2}), \quad &
\sigma^{\phi_1 \vee \phi_2} = \mathsf{bw\_or}(\sigma^{\phi_1},\sigma^{\phi_2})
\end{array}
\]
Using the formula signature, we can also compute the precision
and recall of a formula. In the definitions below, we first define
precision and recall, and then provide the expressions over $\sigma$.
\begin{equation}
\label{eq:precision}
\precision(\phi,P,N) = \frac{ \{ s \mid (s \in P) \wedge (s \models \phi) \} }
                            { \{ s \mid (s \in (P \cup N)) \wedge (s \models \phi) \} }
                     = \frac{ |\{ i \mid i \le |P| \wedge \sigma^\phi_i =1 \} }
                            { len(\sigma^\phi) }
\end{equation}
\begin{equation}
\label{eq:precision}
\recall(\phi,P,N) = \frac{ \{ s \mid (s \in P) \wedge (s \models \phi) \} }
                            { |P| }
                     = \frac{ |\{ i \mid i \le |P| \wedge \sigma^\phi_i =1 \} }
                            { |P| }
\end{equation}
\noindent We present the algorithm for $\learnInvariant$ in two parts: a decision 
tree learner and an enumerative formula constructor with sub-sampling. The decision 
tree learning technique is shown in Algorithm~\ref{alg:dtree}. Let the set of 
concrete atomic predicates be denoted by $\atoms$. Previous work such as \cite{transit}
has investigated the use of enumerative solvers for learning expressions, and while 
they work well in practice, their scalability is limited by the length of the
expression to be learned. The idea in the decision tree learner is to only use the 
enumerative solver for bounded length sub-expressions by partitioning the space of
examples using decision trees. Similar (but different) ideas have  been explored in the 
use of unification-based enumerative solvers in \cite{sygus2}. We essentially only use
the enumerative $\mathsf{InvLearn}$ method on positive and negative example sets of 
size less than some fixed bound $k$. So, in Line~\ref{al:bound}, we check if the size
of $P$ or $N$ exceeds $k$, and if yes, call $\mathsf{DecisionTreeLearner}$ to partition 
the sets into smaller, more manageable subsets (Line~\ref{al:dtlearner}). We do this by 
selecting the $\mathsf{atom}$ with highest precision (Line~\ref{al:entropy}) to split
the set of positive examples, recursively calling the method on the resulting smaller
sized subsets. Otherwise, it proceeds to the enumerative solver in $\mathsf{InvLearn}$ (Line~\ref{al:invlearn}).

In Algorithm~\ref{alg:InvRevision}, we $\mathsf{subSample}$ from positive and negative 
sample sets for our enumerative procedure (Line~\ref{al:subsample}). The broad idea is
to only use the respective subsets $\subP$, $\subN$ of the $P,N$ input to 
$\mathsf{InvLearn}$ to learn an expression. Every atom is checked to see if it
satisfies the precision threshold $\delta$ and has $\recall=1$ (Line~\ref{al:breakatom}).
If yes, we have found a suitable atomic expression and the algorithm breaks out of the loop. If not,
we proceed to the enumerative construction of expressions over $\atoms$ (while 
bounding the overall expression length to be less than $\mathsf{maxInvLength}$). 
Each $\psi$ in Alg.~\ref{alg:InvRevision} corresponds to a new formula of length $\ell$
constructed from subformulas of length $\ell-1$ (with the unary negation operator)
in Line~\ref{al:neg}, or subformulas $\varphi_1$,$\varphi_2$ whose lengths add up to $\ell-1$ (with some binary operator) in Line~\ref{al:op}.  Again,
if the resulting $\psi$ satisfies the precision/recall checks, we break
out of the loop (Line~\ref{al:breakneg}/\ref{al:breakop}). Once all
subformulas of length $\ell$ are exhausted, we increment $\ell$. Whenever
we find a candidate $\psi$, we then check it against the entire sets 
$P$ and $N$, and if it still satisfies the precision/recall checks,
the algorithm returns $\psi$ as the candidate invariant. 
If
it fails to find a candidate invariant with desired precision/recall, $\mathsf{subSample}$ is initiated
again to gather more samples and repeat the entire process to learn a
more accurate invariant (Line~\ref{al:resubsample}).

\section{Experimental Validation}
\label{Sec:experiments}

To show the effectiveness and accuracy of our approach, we evaluated
$\mathsf{InvGen}$ on Promela programs written as test programs for the
Spin model checker\cite{holzmann1997spin}. We address the following
three research questions(RQ) in our experiments: \textbf{RQ1}: Are the
final candidate invariants accurate enough after a limited number of
iterations of $\mathsf{InvGen}$ (i.e., after observing  only a limited
number of reachable states)?  \textbf{RQ2}: Is the invariant tight? \textbf{RQ3}: Is $\mathsf{InvLearn}$
efficiently learning a candidate invariant using the positive/negative
examples discovered by the outer loop in $\mathsf{InvGen}$?
\textbf{RQ4}: Are the final candidate invariants useful in reasoning
about the system safety?
As case studies, we used three classical distributed protocols, and
three distributed programs that have a large number of reachable
states, which cannot be exhaustively explored to compute the
invariant. For evaluating RQ4, we considered two additional case
studies (written in the Ivy programming language) from the Ivy model checker \cite{padon2016ivy}.


\myipara{RQ1} We evaluated RQ1 from two perspectives:
(1) We checked if the invariants were sound, and (2) quantified 
if $\mathsf{InvGen}$ could learn an invariant using
a limited number of observed states. For (1),
we plugged the generated invariants as model checking targets for
Spin/Ivy, and checked if the model checker was able to verify them.
We reiterate that this is a final check on whether the invariants are sound and happens
after $\mathsf{InvGen}$ has concluded. For(2), we measured the
ratio of the number of states observed to the total number of reachable 
states for the program (calculated by hand or obtained through the 
model checker). Table~\ref{tab1} presents the results obtained for each
distributed program. We can see $\phi$ can be learned even when the
number of observed states is small in comparison to the number of
reachable states.

\myipara{RQ2} To
quantify the confidence in the tightness of the learned invariant, we used a
statistical model checking (SMC) approach that gives confidence
guarantees on the learned invariant based on Clopper-Pearson (CP)
interval bounds \cite{confidence,clopper}. The main idea is that the
CP bounds provide the confidence interval for the probability of a
Boolean random variable being true. The Boolean random variable in
question for us is whether the candidate invariant is a true
invariant. Every trace sampled in $\mathsf{InvGen}$ is a trial to
check if a state can invalidate the candidate invariant. Thus, after
$n$ sampled traces, if $r$ are successful trials (i.e. they do not
invalidate the candidate invariant), then we can use this to estimate
a confidence interval on the true probability of the candidate
invariant being valid at a given significance level $(1-\alpha)$. The
formula for CP upper interval bound (i.e., when $r = n$) gives the
confidence interval when all $n$ runs are successful:
$[(\frac{\alpha}{2})^{\frac{1}{n}},1]$.  From this formula, if we need
a significance level of $0.95$ (i.e. $\alpha = 0.05$), we can compute
the required number of trials $n$ to be
$\frac{\log(0.025)}{\log(0.95)} \approx 72$.  The number of final
rounds in Table~\ref{tab1} exceeds 72, which allows us to conclude
that each candidate invariant is likely to be valid with a
significance level of $0.95$.

\myipara{RQ3} To evaluate RQ3, 
we measured the execution time for the
revision process in each round and the number of iterations for which
each candidate invariant survives challenges from sampled traces
(including the final candidate invariant).  We observed that the first
candidate invariant learned by $\mathsf{InvGen}$ 
is subject to immediate revision, but the number of rounds for which
the candidate invariant survives steadily increases for each program.
The final candidate invariants survive for a number of iterations
exceeding the number required for SMC to certify it at the
significance level of $0.95$. We note that RQ3 is validated as the 
average execution time for each revision is quite small.

\myipara{RQ4} To address RQ4, we evaluated if the identified
properties were sufficient to imply the safety property of the
specified system model. We reiterate that the safety properties 
were used after $\mathsf{InvGen}$. In
Table~\ref{tab:learnedinvs}, we enumerate the safety properties
that were specified, and the invariant that we learned automatically
without using these safety properties. For all examples, we 
used the safety invariant $\varphi$ specified in the column denoted 
safety property (corresponding to the LTL property $\box \varphi$). 
For the first two examples, we were also able to show that $\varphi \Rightarrow \psi$, where 
$\psi$ was our learned invariant (shown in the third column). The
Ricart and Agrawala example was taken from the Ivy model checker
repository\cite{padon2016ivy}. For the simple consensus protocol(also written in Ivy), since our invariant expression is 
over more variables than the ones used for specifying the safety 
property, we cannot prove that $\varphi \implies \psi$. However, we 
were able to use the Ivy model checker to confirm that the candidate 
invariant is a true invariant. This leads us to an interesting open question of how we can automatically
choose the set of program variables over which we
wish to learn the invariant so that it is effective
at proving a previously unknown safety property.

\begin{table}[!t]
\caption{Evaluation results on distributed systems with large reachable states. Here, $r$ is the number of sample runs of the program that were used to challenge the invariant after the final revision, and $T_r$ is the average execution time for each revision.}\label{tab1}
\begin{tabular*}{\textwidth}{@{\extracolsep{\fill}}l c c c c r r c}
\toprule
Distributed & Num. & Num. & Num. & $|V|/|R|$  & $r$  & $T_r(s)$  & $\Phi$ \\
Program & Vars. & Reachable & Visited & &  &  & verified? \\
&       & States (R) & States (V) &  &  &   &  \\
\midrule
producer consumer\cite{dijkstra1975guarded} & 4 & 1.03M & 2496 & $0.002$ & 119 & 0.05 & Y \\
distributed lock server\cite{dls} & 4 & 12.2K & 183 & $0.015$ & 92 & 1.89 & Y \\
dining philosophers\cite{dp} & 4 & 492K & 136 & $0.00027$ &102 & 0.98 & Y \\
abp\cite{holzmann1997spin} & 5 & $\infty$ & 517 & 0 & 120 & 0.001 & Y\\
leader election\cite{holzmann1997spin} & 3 & 26K & 63 & $0.0024$ & 95 & 2.50  & Y\\
UPPAAL train/gate\cite{holzmann1997spin} & 7 & 16.8M & 410 & $0.000024$ & 74 & 2.37 & Y\\
\bottomrule
\end{tabular*}
\end{table}

\begin{table}[t]
\caption{Safety properties (obtained from Promela/Spin\cite{holzmann1991promela,holzmann1997spin}) and learned
invariants using $\mathsf{InvGen}$ (Alg.~\ref{alg:InvSynAbs}). We
use the following shorthand for predicates/variable names in the above formulas for brevity.
In UPAAL train/gate, let $v_1 = \mathsf{gate@Add1}, v_2 =
\mathsf{gate@Add2}, v_3 = \mathsf{len(list)}$.
In Rickart-Agarwala, let $h_i = \mathsf{holds(Ni)}$, $req_{ij} =
\mathsf{requested(Ni,Nj)}$, and $rep_{ij} = \mathsf{replied(Ni,Nj)}$.
In simple consensus, $d_{ijk} = \mathsf{decided(Ni, Qj, Vk)}$,
$l_{ij} = \mathsf{leader(Ni,Qj)}$,
$vs_{ij} = \mathsf{votes(Ni,Nj)}$,
$vd_i = \mathsf{voted(Ni)}$.\label{tab:learnedinvs}}
\begin{tabular*}{.99\textwidth}{p{1.5cm}@{\hspace{0.5em}}p{1.7cm}@{\hspace{2em}}p{8cm}}
\toprule
Case Study & Safety Prop. & Learned Invariant \\
\midrule
UPAAL train/gate 
& $\mathsf{(v_1}$ $\mathsf{\lor v_2)}$ $\Rightarrow$ $\mathsf{(v_3 < N)}$
& $((\mathsf{v_3} = 0) \land 
        ((\mathsf{v_1} = 0) \lor 
         (\mathsf{v_2} = 0))
    )$ 
    $\lor ((\mathsf{v_3} > 0) \land (\mathsf{v_3} \leq 3) \land (\mathsf{v_1} = 0))) \lor ((\mathsf{v_3} = 4) \land (\mathsf{v_1} = 0) \land (\mathsf{v_2} = 0)))$ \\
\midrule
Ricart Agrawala
& $ h_1 \land h_2 \Rightarrow \mathsf{(N1 = N2)}$
& $ \mathsf{(N1 \neq N2)} \Rightarrow (\neg req_{22} \land \neg req_{11} \land \neg rep_{22} \land \neg rep_{11} \land 
                 (\neg h_1 \lor \neg h_2))$
\\
\midrule
Simple Consensus
& $(d_{111} \land d_{222})$ $\Rightarrow$ $(v_1 = v_2)$
& $(N1 \neq N2 \land N1 \neq N3 \land N2 \neq N3)$
  $\Rightarrow (\neg l_{31} \land (vd_1 \land \neg d_{311}) \lor (\neg vd_1 \land \neg vs_{11} \land \neg vs_{31})) \lor (l_{31} \land \neg l_{21} \land \neg d_{111})$ \\
\bottomrule
\end{tabular*}

\end{table}

\section{Related Work and Discussion}
\label{sec:related}
\mypara{Related Work} Automatic invariant generation is a fundamental problem in program verification. Early attempts built on first-order theorem
provers such as Vampire~\cite{vampire,tp_v}. but were fundamentally limited by the scalability of
the underlying theorem prover. Subsequent work can be categorized into approaches that guarantee provably correct
invariants, and those that use data-driven approaches to obtain \emph{likely invariants}.
An example of the first family of techniques includes counter-example guided invariant generation (CEGIR). These
approaches rely on a combination of inductive synthesis and verification. Prominent examples
of CEGIR approaches include
the work on implication counter-examples (ICE~\cite{garg2014ice}, ICE-DT~%
\cite{icedt}) and
 FreqHorn~\cite{freqhorn, revisedFre}. FreqHorn identifies inductive
invariants by sampling from a grammar and verifying against an SMT solver. This approach is complementary to
ours, since we sample program executions and learn candidate invariants that effectively distinguish observed from
non-observed states.
These techniques rely on a constraint solver, a theorem prover, or
a model checker to certify that the system satisfies the invariant being proposed.
These approaches can be accelerated in multiple ways, including by obtaining additional information from the verifier
(see, for example, recent work on non-provability information~\cite{nonprovable}), or by using machine learning to
efficiently explore the space of candidate invariants: examples include Code2Inv~\cite{code2inv}, counter-example guided
neural synthesis~\cite{nn}, and ACHAR~\cite{almostcorrect}.
The cost of model checking and the rapid growth of the space of possible invariants limit the complexity of systems
to which these techniques can be applied. On the other hand, our approach in $\mathsf{InvGen}$ relies on execution traces and
revision events, and can thus be freely applied to complex systems.

Work on data-driven invariant generation was pioneered by Daikon~\cite{ernst2007daikon}, which maintained a pool of
candidate invariants that was iteratively pruned as new reachable states were observed. The main limitation is its sensitivity to the initial pool of candidate invariants, and its inability to learn invariants with
complex Boolean structure, including properties with negation. More recent examples of this approach include NumInv~%
\cite{numinv}, DistAI~\cite{distai}, and DuoAI~\cite{yao2022duoai}. DistAI attempts to verify distributed protocols
expressed using the Ivy language~\cite{padon2016ivy}. Although this approach also relies on observations of system
executions, unlike $\mathsf{InvGen}$, it is notably guided by a target property being verified, and this might cause the system to
loop forever in cases where the target property fails to hold.


\noindent\textbf{Discussion and Future Work}
$\mathsf{InvGen}$ is an automated and practical framework for observing non-deterministic system behaviors and generating accurate
system properties. The approach is data-driven, leveraging counter-examples and speculative negative states to guide the
refinement process. The learning model is template-free and decision-tree-based. 
Future work includes multiple promising directions. First, exploring alternative speculative sampling techniques
enhances the convergence of the learning process. Second, the decision tree learning method could be further
investigated, including the separator selection and the exploit-exploration trade-off.


\bibliographystyle{splncs04}
\bibliography{bib}


\end{document}